\documentclass[11pt]{article}

\usepackage[a4paper,margin=1in]{geometry}
\usepackage{amsmath,amssymb,amsthm,bm}
\usepackage{booktabs}
\usepackage{array}
\usepackage{graphicx}
\usepackage{hyperref}
\usepackage{tikz}
\usepackage{enumitem}
\usetikzlibrary{arrows.meta,positioning}

\hypersetup{colorlinks=true,linkcolor=blue,citecolor=blue,urlcolor=blue}

\newcommand{\ket}[1]{|#1\rangle}
\newcommand{\bra}[1]{\langle #1|}
\newcommand{\Enc}{\operatorname{Enc}}
\newcommand{\Good}{\mathcal G}
\newcommand{\Key}{\mathcal K}
\newcommand{\Data}{\mathcal D}

\newcommand{\Prb}{\mathbb P}

\newtheorem{assumption}{Assumption}
\newtheorem{lemma}{Lemma}
\newtheorem{proposition}{Proposition}
\newtheorem{corollary}{Corollary}

\title{Cipher-Structure-Aware Variational Quantum Cryptanalysis:\\
A Reversible Public-Diagnostic Framework}

\author{Xi Li\\
\small Henan University, Kaifeng, China}

\date{}

\begin{document}
\maketitle

\begin{abstract}
We introduce a cipher-structure-aware variational framework for key recovery from public data in reduced block ciphers. Public S-box differential constraints and candidate-key-dependent inverse-round residuals are converted into phase separators through reversible compute--phase--uncompute circuits. The principal separator is a hybrid DDT construction that combines pairwise and synchronized multi-pair differential constraints before residual refinement. The construction is specified by public functions and does not require a key-indexed phase table in the physical circuit. We formalize the public-instance model, prove one-sided completeness of the structural diagnostics, and derive a conditional lower-tail bound for false-key diagnostic gaps under an explicitly stated sub-Gaussian moment condition. This statistical statement concerns the diagnostic landscape, not convergence of a finite-depth QAOA optimizer. On ten independent random three-pair instances of 16-bit simplified AES, an 18-layer Hybrid-DDT schedule reaches a mean public-consistent-key probability of $0.529785$, with values from $0.388892$ to $0.770937$, ranks the correct key first in every run, and has largest observed false-key probability $0.008723$. The result is a reduced-cipher mechanism study, not an attack on full AES.
\end{abstract}

\section{Introduction}

Given public plaintext--ciphertext pairs, symmetric-key recovery can be formulated as a reversible public-verification problem. A candidate key is accepted when the reversible cipher reproduces the published ciphertexts. Quantum analyses of symmetric-key primitives have emphasized the cost of reversible verification, while recent variational studies have explored heuristic state preparation for cryptographic search problems~\cite{BonnetainAES,WangVQAA,AizpuruaVQAA,AizpuruaTN}. The construction uses public round structure to shape the candidate-key register before this final test: differential properties of the S-box and candidate-key-dependent inverse checkpoints supply intermediate phase information.

Variational quantum algorithms and alternating-operator ansatze provide a setting in which these diagnostics can be interleaved with key-register mixing~\cite{FarhiQAOA,HadfieldAOA,CerezoVQA}. Three objects must be distinguished: the public diagnostic, the reversible circuit that evaluates it, and the final public-consistency test. A diagnostic phase is meaningful only when it is computable from public data in superposition and when the reported overlap is defined on the candidate-key register.

We instantiate this approach for a reduced AES-like block cipher. The central object is a public diagnostic
\begin{equation}
    f_r:\Key\longrightarrow\mathcal Z_r,
    \qquad e_r(k)=h_r(f_r(k)),
    \label{eq:intro-diagnostic}
\end{equation}
where $f_r$ is computed from a candidate key and public plaintext--ciphertext pairs, and $h_r$ maps the diagnostic code to a phase energy. The diagnostic is not a correctness predicate; it is an intermediate phase function used to prepare the key register for final public verification.

The contribution is fivefold.
\begin{enumerate}[leftmargin=*]
    \item We specify a public-instance model that separates synthetic-key generation from the data available to the recovery procedure and defines the recovery metric on the key register.
    \item We construct a hybrid S-box DDT separator from pairwise and synchronized joint components, together with candidate-key-dependent inverse-round residual separators. The DDT construction follows a forward-candidate path and does not use a tautological ciphertext pullback.
    \item We prove one-sided diagnostic completeness and a conditional lower-tail bound for false keys. The analysis quantifies diagnostic selectivity without asserting convergence of finite-depth QAOA.
    \item We give the reversible primitive $U_f^{\dagger}P_fU_f$ and distinguish its physical circuit realization from the compressed 16-bit statevector replay used for numerical auditing.
    \item We report a ten-instance Hybrid-DDT S-AES benchmark, a pairwise public-pair reference, a no-NPP pairwise/Hybrid ablation, and a gate-level S-DES validation.
\end{enumerate}

The experiments concern a reduced cipher and should be read as a mechanism study. They do not establish key recovery for full AES. The diagnostic assumptions are stated in a form that can be tested on other reduced-cipher ensembles.

\section{Public Instance and Evaluation Model}
\label{sec:model}

\subsection{Instance generation}

Let the key space be
\begin{equation}
    \Key=\{0,1\}^{n},
\end{equation}
and let $\Enc_k$ be a public block-cipher implementation with block length $b$. A public-instance ensemble is specified by a distribution for the plaintext vector $\mathbf P=(P_1,\ldots,P_m)$. A synthetic instance first samples a secret key
\begin{equation}
    K\leftarrow\operatorname{Unif}(\Key),
    \label{eq:key-sampling}
\end{equation}
and then chooses public plaintexts independently of $K$. The canonical random-text model takes
\begin{equation}
    P_j\stackrel{\mathrm{i.i.d.}}{\sim}\operatorname{Unif}(\{0,1\}^{b}),
    \qquad P_j\mathrel{\perp} K,
    \label{eq:plaintext-sampling}
\end{equation}
although the analysis also applies after conditioning on a fixed public plaintext design. The experiment publishes
\begin{equation}
    C_j=\Enc_K(P_j),
    \qquad
    \Data=\{(P_j,C_j)\}_{j=1}^{m}.
    \label{eq:public-data}
\end{equation}
After generation, $K$ is not supplied to the variational procedure. It is used only by the experimenter to create the public ciphertexts and to audit the result.

The public-consistent key set is
\begin{equation}
    \Good(\Data)=
    \{k\in\Key:\Enc_k(P_j)=C_j\text{ for all }j\}.
    \label{eq:good-set}
\end{equation}
By construction, $K\in\Good(\Data)$, so this set is nonempty. It can contain more than one key for a reduced cipher or for too few public pairs. The algorithm therefore targets the set $\Good$, not an inaccessible label for one particular key.

When $\Key\setminus\Good(\Data)$ is nonempty, an error key is drawn for analysis as
\begin{equation}
    \widetilde K\mid\Data
    \leftarrow\operatorname{Unif}\bigl(\Key\setminus\Good(\Data)\bigr).
    \label{eq:false-key-sampling}
\end{equation}
This conditional distribution is used only to state and test the statistical model; it is not available to the recovery procedure.
The ensemble statements below are restricted to instances for which $\Key\setminus\Good(\Data)$ is nonempty almost surely; equivalently, one may condition the ensemble on this event.

\subsection{Key-register objective}

For a prepared key-register state
\begin{equation}
    \ket{\psi(\Theta)}=\sum_{k\in\Key}\alpha_k(\Theta)\ket{k},
\end{equation}
we define
\begin{equation}
    p_{\Good}(\Theta)=
    \sum_{k\in\Good}|\alpha_k(\Theta)|^2.
    \label{eq:pgood-qip}
\end{equation}
When $\Good=\{K\}$, this is the true-key probability. Otherwise it is the probability of obtaining a key that passes all public pairs. It is not a ciphertext-register overlap.

The recovery procedure consists of three stages:
\begin{enumerate}[leftmargin=*]
    \item prepare $\ket{+}^{\otimes n}$ on the candidate-key register;
    \item apply public structural phase separators and key-register mixers;
    \item measure the key register and verify the observed candidates with the public equations in Eq.~\eqref{eq:good-set}.
\end{enumerate}

The training objective uses only the public pairs and the candidate-key state; it does not use the hidden generator key $K$ or an exact-key label. In the physical construction, a ciphertext-space projector can be estimated from the public ciphertext basis states. In the compressed replay, the deterministic encryption computation is pulled back to a key-register mask, which is mathematically equivalent to the public-consistency projector and is used only as a simulation shortcut. The set $\Good$ is therefore a public evaluation set, not hidden side information. Exact summation over the reduced key register removes shot noise from the numerical comparison.

\section{Cipher-Structure-Aware Phase Separators}
\label{sec:diagnostics}

We instantiate the construction on the 16-bit simplified AES (S-AES) structure~\cite{NISTAES,DaemenRijmen,Stallings}. The state consists of four 4-bit nibbles. The round function contains a public 4-bit S-box, a ShiftRows-like permutation, MixColumns over $\mathrm{GF}(2^4)$, and XOR with round keys. The key schedule expands a 16-bit master key into $K_0,K_1,K_2$. The encryption equations are
\begin{align}
    S_0&=P\oplus K_0,\\
    S_1&=\operatorname{MC}(\operatorname{SR}(\operatorname{SN}(S_0)))\oplus K_1,\\
    C&=\operatorname{SR}(\operatorname{SN}(S_1))\oplus K_2.
    \label{eq:saes-qip}
\end{align}
The S-box is the public permutation
\begin{equation}
    S=(9,4,A,B,D,1,8,5,6,2,0,3,C,E,F,7).
\end{equation}

\subsection{Forward-candidate DDT separator}

For two public pairs $(P_a,C_a)$ and $(P_b,C_b)$, the final AddRoundKey cancels in the ciphertext difference. The public final-round S-box output difference is therefore
\begin{equation}
    \Delta y_{ab}=\operatorname{SR}^{-1}(C_a\oplus C_b).
    \label{eq:public-dy}
\end{equation}
For a candidate key $k$, the forward prefix computes
\begin{equation}
    X_j^F(k)=
    \operatorname{MC}(\operatorname{SR}(\operatorname{SN}(P_j\oplus K_0(k))))\oplus K_1(k),
    \label{eq:forward-prefix}
\end{equation}
which is the input to the final S-box. The candidate input-difference word is
\begin{equation}
    \Delta x_{ab}(k)=X_a^F(k)\oplus X_b^F(k).
    \label{eq:candidate-dx}
\end{equation}
For nibble $q$, write $u_{ab,q}(k)=[\Delta x_{ab}(k)]_q$ and $v_{ab,q}=[\Delta y_{ab}]_q$ for the corresponding four-bit components.
This forward-candidate construction evaluates the candidate input difference from the encryption prefix. It does not define the differential test by pulling the public ciphertext difference through a candidate-dependent inverse map.

Let
\begin{equation}
    N_{uv}=
    \left|\left\{x\in\{0,1\}^{4}:
    S(x)\oplus S(x\oplus u)=v\right\}\right|
\end{equation}
denote the public S-box DDT count. This is the standard difference-count object used in differential cryptanalysis~\cite{BihamShamir1993,Nyberg1994}. For a nibble with candidate input difference $u$ and public output difference $v$, we use
\begin{equation}
 h_{\mathrm{DDT}}(u,v)=
\begin{cases}
-\log(N_{uv}/16),&N_{uv}>0,\\
\zeta,&N_{uv}=0.
\end{cases}
\label{eq:ddt-energy-qip}
\end{equation}
Here $\log$ denotes the natural logarithm. The zero-count penalty $\zeta$ is finite in the numerical implementation. The code standardizes each component over this finite local score alphabet before optimization; this is a positive affine reparameterization of the component energy, absorbed by its variational angle up to a key-independent phase. The principal separator is hybrid: it retains the selected pairwise nibble components and adds synchronized multi-pair components. For three public pairs, the implementation uses twelve pairwise nibble components and four joint nibble components, for sixteen DDT components in total. The joint terms are intended to reduce differential-equivalent false peaks; neither the pairwise nor joint DDT support condition replaces exact verification.

For a synchronized joint component, the ordinary count is replaced by the public joint count
\begin{equation}
    N^{(J)}_{u_1,u_2;v_1,v_2}
    =
    \left|\left\{x\in\{0,1\}^{4}:
    \begin{array}{l}
    S(x)\oplus S(x\oplus u_1)=v_1,\\
    S(x)\oplus S(x\oplus u_2)=v_2
    \end{array}\right\}\right|,
\end{equation}
with the same score map after replacing $N_{uv}$ by $N^{(J)}_{u_1,u_2;v_1,v_2}$. Both counts are over the fixed 16-element S-box domain; they are public local alphabets, not key-space tables.
In the three-pair construction, $(u_1,v_1)$ and $(u_2,v_2)$ are the pairwise transitions relative to a common reference pair, while the third pairwise difference is retained in the ordinary pairwise bank. Thus the joint bank couples the public differences without introducing a key-space lookup table.

\subsection{Inverse-round residual separator}

For a selected checkpoint $t$, let $R_{t,k}$ be a candidate-key-dependent inverse diagnostic map on the relevant cipher-state space. It is built from the public S-AES modules and the round keys derived from $k$. The synchronized residual is
\begin{equation}
    \delta_{j,t}(k)=
    R_{t,k}(\Enc_k(P_j))\oplus R_{t,k}(C_j).
    \label{eq:residual-qip}
\end{equation}
For a public-consistent key, the two arguments coincide and $\delta_{j,t}(k)=0$. The selected maps include the inverse final S-box, the complete inverse final round, and a two-round inverse checkpoint. The residual phase may use nibblewise Hamming weights or a normalized smooth energy. Because the maps depend on the candidate round keys, they cannot be replaced by a fixed public intermediate state.

If $R_{t,k}$ is injective on the relevant state space, then $\delta_{j,t}(k)=0$ is equivalent to $\Enc_k(P_j)=C_j$ for that branch. The inverse checkpoints used here are reversible for each fixed $k$, so this equivalence holds for the stated maps. In that case the zero residual is a re-encoded verification discrepancy, whereas its nonzero Hamming profile can still provide a useful phase landscape. If a diagnostic map is partial or noninjective, only the one-sided implication needed below is retained. Thus the residual is not treated as an independent correctness theorem; exact public verification remains the final decision rule.

\subsection{Optional residual-shell aggregation}

An optional global residual Hamming-distance term aggregates the selected residual words. If $d_{j,t}(k)$ is the Hamming distance of a residual word, the aggregate is
\begin{equation}
    d(k)=\sum_{j,t}d_{j,t}(k).
    \label{eq:aggregate-distance}
\end{equation}
The exact-zero shell is assigned the strongest attractive phase, near shells receive a weaker phase, and distant shells are penalized. This term is optional and is not required for the DDT--residual construction. No number-partitioning oracle is assumed by the core method.

\section{Reversible Circuit Construction}
\label{sec:physical}

\subsection{Diagnostic sandwich}

Each public diagnostic is implemented by a compute--phase--uncompute circuit. For a diagnostic $f_r$ and energy map $h_r$, define a reversible computation
\begin{equation}
    U_{f_r}\ket{k}\ket{0}=\ket{k}\ket{f_r(k)}.
\end{equation}
The phase operation acts only on the diagnostic workspace:
\begin{equation}
\ket{k}\ket{0}
\xrightarrow{U_{f_r}}
\ket{k}\ket{f_r(k)}
\xrightarrow{P_{h_r}}
e^{-i\gamma h_r(f_r(k))}\ket{k}\ket{f_r(k)}
\xrightarrow{U_{f_r}^{\dagger}}
e^{-i\gamma h_r(f_r(k))}\ket{k}\ket{0}.
\label{eq:sandwich-qip}
\end{equation}
The induced key-register operator is consequently
\begin{equation}
    U_r(\gamma)=
    \sum_{k\in\Key}e^{-i\gamma e_r(k)}\ket{k}\bra{k}.
    \label{eq:induced-diagonal}
\end{equation}
Equation~\eqref{eq:induced-diagonal} is a consequence of reversible computation; it is not a prescription to load a classical array of $2^n$ phases.

The DDT circuit computes the key schedule and the forward prefix in Eq.~\eqref{eq:forward-prefix}, obtains the candidate difference, compares it with a public output difference, applies a phase controlled by a constant-size DDT code, and uncomputes all work registers. The residual circuit computes candidate encryption branches, applies the same inverse checkpoint to the candidate and public branches, computes residual bits or counters, applies the phase, and uncomputes. The S-box is a fixed small permutation, while the linear layers are reversible maps over $\mathrm{GF}(2^4)$. This compute--phase--uncompute pattern follows the standard logical-reversibility model and can be decomposed into elementary quantum gates~\cite{Bennett1973,Barenco1995}.

\subsection{Staged QAOA}

Let $\mathcal R_\ell$ be the diagnostic components used at layer $\ell$. The key-register state is
\begin{equation}
    \ket{\psi(\Theta)}=
    \prod_{\ell=1}^{p}
    U_M(\beta_\ell,b_\ell)
    \exp\left[-i\sum_{r\in\mathcal R_\ell}
        \gamma_{\ell r}H_{\ell r}\right]
    \ket{+}^{\otimes n},
    \label{eq:qaoa-qip}
\end{equation}
where
\begin{equation}
    U_M(\beta,b)=
    \exp\left[-i\beta\left(\sum_iX_i+b\sum_i|1\rangle\langle1|_i\right)\right].
    \label{eq:biased-mixer}
\end{equation}
Setting $b=0$ gives the standard transverse-field mixer. Component phases in one layer are applied serially. After the diagnostic workspaces are uncomputed, the component operators are diagonal on the key register and their product is exactly the exponential of the sum in Eq.~\eqref{eq:qaoa-qip}; no intermediate measurement is used.

The product is ordered in time from right to left, so the rightmost factor acts first. The layer index is a circuit index, not a claim that all diagnostics are evaluated simultaneously. In the staged implementation, consecutive blocks may use different diagnostic families and different numbers of layers.

The principal staged schedule is
\begin{equation}
    \operatorname{DDT}^{\rm hybrid}_{12}\longrightarrow
    \operatorname{residual}_{3}\longrightarrow
    \operatorname{GHD}_{3},
    \label{eq:main-schedule-qip}
\end{equation}
where the last block is an optional residual-shell aggregation. For three public pairs, the Hybrid DDT block has twelve pairwise nibble components and four synchronized joint components. The subscript $12$ denotes twelve QAOA layers, not twelve components. The angles are optimized using a public objective, and measured candidates are verified independently.

\begin{figure}[t]
    \centering
    \includegraphics[width=0.92\linewidth]{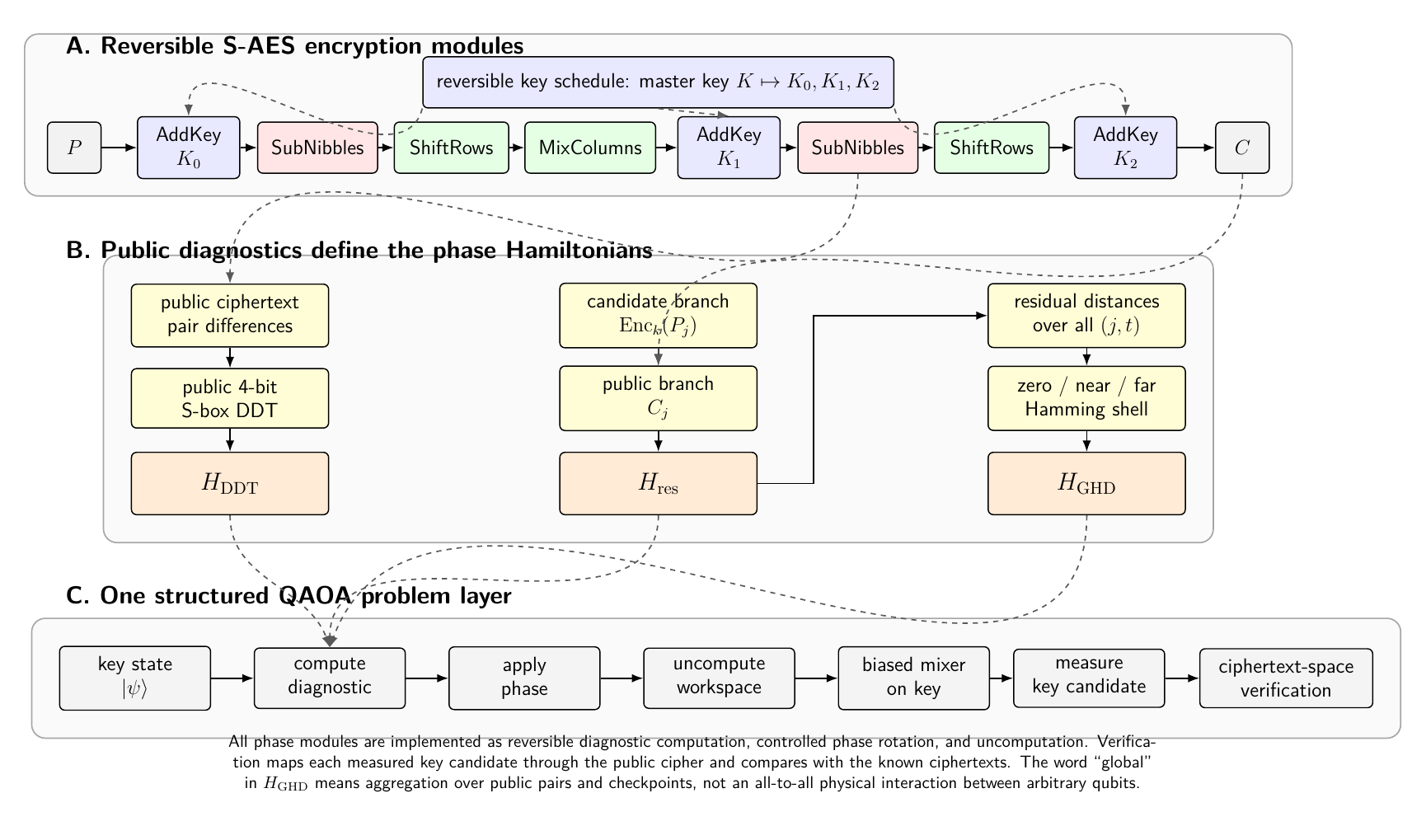}
    \caption{Structure-aware key-register construction. Cipher modules are computed reversibly from a candidate key and public data. A diagnostic phase is inserted between computation and uncomputation, after which the mixer acts only on the clean key register. The effective 16-bit replay stores the induced diagonal for auditing; the physical route uses the reversible sandwiches in Eq.~\eqref{eq:sandwich-qip}.}
    \label{fig:module-map-qip}
\end{figure}

\subsection{Numerical replay versus physical circuit}

The 16-bit numerical benchmark materializes the induced Hybrid-DDT, residual, and residual-shell arrays over all $2^{16}$ candidate keys and uses them as Qiskit diagonal phase blocks on a 16-qubit effective key register. This gives an exact statevector evaluation of the reduced key-register action and removes shot noise. It remains a classical compressed replay, not a full gate-level S-AES oracle.

The physical implementation replaces each diagonal block by Eq.~\eqref{eq:sandwich-qip}. A full statevector simulation of that route must include the key, encryption, key-schedule, diagnostic, and comparison workspaces. A bare key-plus-state verification model already requires at least 32 logical qubits, corresponding to 64 GiB in complex double precision before simulator overhead. This is a limitation of classical statevector simulation, not a requirement that a physical quantum device provide that classical memory.

\section{Deterministic Properties and Conditional Separation}
\label{sec:theory}

\subsection{Diagnostic energy and deterministic completeness}

Fix a public instance and write $e_r(k)=h_r(f_r(k))$ for the score of component $r$. Let $L$ be the number of selected components. Two notions must be kept separate. A DDT transition is \emph{feasible} when its count is nonzero, whereas the numerical score in Eq.~\eqref{eq:ddt-energy-qip} is a likelihood-like ranking energy and is generally positive even for a feasible transition. Accordingly, define the support band
\begin{equation}
\chi_r(k)=
\begin{cases}
1,&\text{if the selected DDT count for component $r$ is nonzero},\\
0,&\text{otherwise},
\end{cases}
\label{eq:ddt-support}
\end{equation}
and
\begin{equation}
\mathcal B_{\rm supp}=
\left\{k\in\Key:
\begin{array}{l}
\chi_r(k)=1\text{ for every selected DDT component }r,\\
\delta_{j,t}(k)=0\text{ for every selected residual component }(j,t)
\end{array}\right\}.
\label{eq:band-qip}
\end{equation}

\begin{lemma}[One-sided diagnostic completeness]
For the S-AES diagnostics above,
\begin{equation}
    \Good\subseteq\mathcal B_{\rm supp}.
    \label{eq:completeness-qip}
\end{equation}
\end{lemma}

\begin{proof}
Let $g\in\Good$. The final AddRoundKey cancels in each public ciphertext difference, so Eq.~\eqref{eq:public-dy} gives the actual final S-box output difference. The forward prefix computed with $g$ gives the actual final S-box input difference; the corresponding ordinary DDT entry is therefore nonzero. For a joint component, the same actual S-box input realizes the two observed transitions simultaneously, so the joint count is also nonzero. For every residual checkpoint, $\Enc_g(P_j)=C_j$, so both arguments in Eq.~\eqref{eq:residual-qip} are identical and $\delta_{j,t}(g)=0$. Therefore $g\in\mathcal B_{\rm supp}$.
\end{proof}

The converse is not claimed. Differential-equivalent and checkpoint-equivalent false keys may remain in $\mathcal B_{\rm supp}$, and a nonzero DDT count does not imply zero value of Eq.~\eqref{eq:ddt-energy-qip}. Exact public verification remains the correctness predicate.

For the statistical statements, define the total diagnostic energy and an instance-dependent reference by
\begin{equation}
    e_{\Data}(k)=\sum_{r=1}^{L}e_r(k),
    \qquad
    a_{\Data}=\max_{g\in\Good(\Data)}e_{\Data}(g),
    \qquad
    \Gamma_{\Data}(k)=e_{\Data}(k)-a_{\Data}.
    \label{eq:energy-qip}
\end{equation}
The reference $a_{\Data}$ is introduced only for analysis and is never supplied to the algorithm. Subtracting it from all key energies changes the phase by a global phase. Every public-consistent key satisfies $\Gamma_{\Data}(g)\leq0$, and at least one such key has zero reference gap. This definition remains valid when $|\Good|>1$; it avoids the invalid assumption that every consistent key minimizes every component energy.

The aggregate $e_{\Data}$ is a fixed analysis score with unit component weights. Fixed positive weights can be accommodated provided that the corresponding rescaling is included in the constants below. It is not identified with the trained QAOA phase function: the learned angles are component- and layer-dependent and may have either sign.

\subsection{Sub-Gaussian public-diagnostic model}

The statistical claim below is an ensemble statement. Let $\mathbb E_{\rm pub}$ and $\Prb_{\rm pub}$ denote expectation and probability over the key and plaintext sampling in Eqs.~\eqref{eq:plaintext-sampling} and~\eqref{eq:key-sampling}. Let $\mathbb E_{\rm ens}$ and $\Prb_{\rm ens}$ additionally include the conditional false-key sampling in Eq.~\eqref{eq:false-key-sampling}. The constants in the next assumption are properties of the diagnostic family and the chosen public-instance ensemble, not QAOA parameters. For the fixed-text benchmark, the same statement is understood conditionally on the displayed plaintext vector rather than as an i.i.d.-plaintext claim.

\begin{assumption}[Sub-Gaussian false-key gap]
There exist $\mu>0$, $\nu>0$, and $\lambda_0>0$ such that for every $0<\lambda\leq\lambda_0$,
\begin{equation}
    \mathbb E_{\rm ens}\left[
\exp\left(-\lambda\left(\Gamma_{\Data}(\widetilde K)-\mu L\right)\right)
\right]
\leq
\exp\left(\frac{\lambda^2\nu L}{2}\right).
\label{eq:mgf-qip}
\end{equation}
\end{assumption}

This is a one-sided sub-Gaussian moment condition on the accumulated false-key gap. It does not require independent diagnostic components; independence with suitable per-component bounds is only one sufficient route to it. The condition can fail when false keys systematically imitate the low-score diagnostic band. It is therefore a hypothesis to be calibrated and tested on held-out instances, not a property inferred from the numerical optimization.

\begin{proposition}[False-key low-energy tail]
Under Assumption~1, for every $\tau<\mu$,
\begin{equation}
    \Prb_{\rm ens}\left[\Gamma_{\Data}(\widetilde K)\leq\tau L\right]
    \leq
    \exp[-L I(\tau)],
    \label{eq:tail-qip}
\end{equation}
where
\begin{equation}
    I(\tau)=
    \sup_{0<\lambda\leq\lambda_0}
    \left\{\lambda(\mu-\tau)-\frac{\lambda^2\nu}{2}\right\}.
    \label{eq:rate-qip}
\end{equation}
If the unconstrained optimizer is admissible, then $I(\tau)=(\mu-\tau)^2/(2\nu)$.
\end{proposition}

\begin{proof}
For any admissible $\lambda$, Markov's inequality gives
\begin{align}
\Prb_{\rm ens}[\Gamma_{\Data}(\widetilde K)\leq\tau L]
 &=\Prb_{\rm ens}[-(\Gamma_{\Data}(\widetilde K)-\mu L)\geq(\mu-\tau)L]\nonumber\\
&\leq e^{-\lambda(\mu-\tau)L}
  \mathbb E_{\rm ens}[e^{-\lambda(\Gamma_{\Data}(\widetilde K)-\mu L)}]\nonumber\\
&\leq\exp\left[-\lambda(\mu-\tau)L+\frac{\lambda^2\nu L}{2}\right].
\end{align}
Optimizing over $\lambda$ proves Eq.~\eqref{eq:tail-qip}.
\end{proof}

\begin{corollary}[Expected false-key count]
Let $N_{\rm false}(\tau)=|\{k\in\Key\setminus\Good(\Data):\Gamma_{\Data}(k)\leq\tau L\}|$ and set $q_\tau=\exp[-L I(\tau)]$. If, for every fixed $k$ with positive false-key probability,
\begin{equation}
    \Prb_{\rm pub}\left[\Gamma_{\Data}(k)\leq\tau L\mid k\notin\Good(\Data)\right]
    \leq q_\tau,
    \label{eq:uniform-fixed-key-tail-qip}
\end{equation}
then
\begin{equation}
    \mathbb E_{\rm pub}[N_{\rm false}(\tau)]
    \leq 2^n q_\tau.
    \label{eq:false-count-qip}
\end{equation}
If $|\Good(\Data)|=g$ almost surely and the tail bound of Proposition~1 holds for the uniformly sampled false key in Eq.~\eqref{eq:false-key-sampling}, then the sharper identity
\begin{equation}
    \mathbb E_{\rm pub}[N_{\rm false}(\tau)]
    =(2^n-g)\Prb_{\rm ens}\left[\Gamma_{\Data}(\widetilde K)\leq\tau L\right]
    \leq (2^n-g)q_\tau
    \label{eq:false-count-fixed-good-qip}
\end{equation}
holds. In particular, $LI(\tau)>n\ln 2$ makes the first bound smaller than one. Without either the fixed-index uniformity in Eq.~\eqref{eq:uniform-fixed-key-tail-qip} or a deterministic false-key count, an average bound for one uniformly sampled false key cannot be multiplied by a random number of false keys.
\end{corollary}

\begin{proof}
Write $N_{\rm false}(\tau)$ as a sum of indicators over $k\in\Key$. For each term, multiply the conditional bound in Eq.~\eqref{eq:uniform-fixed-key-tail-qip} by $\Prb_{\rm pub}[k\notin\Good(\Data)]\leq1$, and then use linearity of expectation to obtain Eq.~\eqref{eq:false-count-qip}. If $|\Good|=g$ is deterministic, conditioning on $\Data$ and sampling $\widetilde K$ uniformly over the $2^n-g$ false keys gives the identity in Eq.~\eqref{eq:false-count-fixed-good-qip}.
\end{proof}

\subsection{Hamming-shell specialization}

Suppose $B$ selected residual bits for a false key are unbiased and independent in an idealized null model. Then
\begin{equation}
    W=|\delta(\widetilde K)|_1\sim\operatorname{Binomial}(B,1/2).
\end{equation}
For $0\leq\rho<1/2$,
\begin{equation}
    \Prb[W\leq\rho B]
    =2^{-B}\sum_{i=0}^{\lfloor\rho B\rfloor}\binom{B}{i}
    \leq2^{-B[1-H_2(\rho)]},
    \label{eq:hamming-qip}
\end{equation}
where
\begin{equation}
    H_2(\rho)=-\rho\log_2\rho-(1-\rho)\log_2(1-\rho)
\end{equation}
is the binary entropy, with the usual convention $H_2(0)=0$. This bound concerns a shell event, not the full smooth reward used by every numerical configuration. It illustrates the source of exponential suppression in the idealized model; dependence and cipher-specific algebraic correlations must be measured on held-out data.

\subsection{From landscape separation to QAOA probability}

The lower-tail result concerns the diagnostic landscape and does not, by itself, establish that a finite-depth QAOA optimizer reaches its low-energy region. For a fixed public instance and $\eta>0$, consider the idealized nonunitary filter
\begin{equation}
    \ket{\phi_\eta}=
    \frac{1}{\sqrt{Z_\eta}}
    \sum_{k\in\Key}e^{-\eta\Gamma_{\Data}(k)/2}\ket{k}.
    \label{eq:filter-qip}
\end{equation}
Here $Z_\eta=\sum_{k\in\Key}e^{-\eta\Gamma_{\Data}(k)}$ normalizes the state.
\begin{proposition}[Ideal-filter bound]
For a fixed public instance, if every $g\in\Good$ satisfies $\Gamma_{\Data}(g)\leq0$ and every false key satisfies $\Gamma_{\Data}(k)\geq\Delta>0$, then
\begin{equation}
    p_{\Good}(\phi_\eta)
    \geq
    \frac{|\Good|}{|\Good|+(2^n-|\Good|)e^{-\eta\Delta}}.
    \label{eq:filter-bound-qip}
\end{equation}
\end{proposition}

\begin{proof}
Let $A=\sum_{g\in\Good}e^{-\eta\Gamma_{\Data}(g)}$ and $B=\sum_{k\notin\Good}e^{-\eta\Gamma_{\Data}(k)}$. The hypotheses give $A\geq|\Good|$ and $B\leq(2^n-|\Good|)e^{-\eta\Delta}$. Since $p_{\Good}(\phi_\eta)=A/(A+B)$ is increasing in $A$ and decreasing in $B$, Eq.~\eqref{eq:filter-bound-qip} follows.
\end{proof}

The global score-gap condition is stronger than the average lower-tail assumption; the latter does not imply that no false key lies below the threshold. The filter is a reference nonunitary state used for analysis, not an additional operation claimed to be implemented by the QAOA circuit.

\begin{proposition}[Variational transfer]
Let $\Pi_{\Good}$ project onto the public-consistent key space. After choosing the irrelevant global phase, if a trained QAOA state satisfies
\begin{equation}
    \|\ket{\psi(\Theta)}-\ket{\phi_\eta}\|_2\leq\varepsilon,
\end{equation}
then
\begin{equation}
    p_{\Good}(\Theta)
    \geq
    \left[\max\left\{0,\sqrt{p_{\Good}(\phi_\eta)}-\varepsilon\right\}\right]^2.
    \label{eq:transfer-qip}
\end{equation}
\end{proposition}

\begin{proof}
The projector is a contraction, so
\begin{equation}
\left|\|\Pi_{\Good}\ket{\psi(\Theta)}\|_2-
\|\Pi_{\Good}\ket{\phi_\eta}\|_2\right|
\leq\|\ket{\psi(\Theta)}-\ket{\phi_\eta}\|_2\leq\varepsilon.
\end{equation}
Taking the positive part and squaring gives Eq.~\eqref{eq:transfer-qip}.
\end{proof}

The two conditional results address different layers of the construction. Equation~\eqref{eq:tail-qip} controls the frequency of low-score false keys in an ensemble, while Eq.~\eqref{eq:transfer-qip} quantifies the consequence of approximating a particular reference filter. Neither result is a finite-depth QAOA convergence theorem. The measured finite-depth probability remains an empirical quantity. Moreover, the analysis score $e_{\Data}$ is a fixed diagnostic aggregate; the learned component angles need not be equal, positive, or proportional to its coefficients. No identification between the reference filter and the trained ansatz is assumed.

\section{Circuit Resources and Model Assumptions}
\label{sec:resources}

Let $w$ denote the number of workspace qubits required by one diagnostic computation. The physical circuit uses $n+w$ logical qubits, with workspace reuse across phase modules. Its width is not $2^n$. For fixed precision and a reversible implementation of the cipher, one diagnostic sandwich has gate cost proportional to the cost of computing the cipher and the diagnostic, plus the cost of a small counter, comparator, or controlled phase. A $p$-layer schedule therefore has cost proportional to the number of selected diagnostics and layers, up to the usual synthesis and routing overhead.

More explicitly, if $C(f_r)$ denotes the reversible gate cost of computing diagnostic $f_r$, $C(P_r)$ the cost of its controlled phase, and $C_M$ the cost of one mixer layer, then a schedule with component set $\mathcal R_\ell$ has schematic gate cost
\begin{equation}
    C_{\rm circ}
    =
    O\left(
    \sum_{\ell}
    \left[
    \sum_{r\in\mathcal R_\ell}
    \bigl(2C(f_r)+C(P_r)\bigr)
    +C_M
    \right]\right),
    \label{eq:circuit-cost-qip}
\end{equation}
with width $n+\max_r w_r$ when workspaces are reused. The constants depend on reversible arithmetic, phase precision, and hardware routing. Thus the construction avoids a key-indexed table in circuit width, but this circuit-size statement is not a polynomial-time key-recovery theorem. Training may require many circuit evaluations, and direct measurement requires $O(1/p_{\Good})$ samples in expectation before a public-consistent key is observed.

The reduced benchmark and the physical model have different roles:
\begin{center}
\begin{tabular}{@{}>{\raggedright\arraybackslash}p{0.25\linewidth}>{\raggedright\arraybackslash}p{0.68\linewidth}@{}}
\toprule
model & role \\
\midrule
compressed 16-bit replay & exact audit of the induced key-register dynamics; enumerates $2^{16}$ keys classically to construct phase arrays \\
reversible diagnostic circuit & physical specification using S-box permutations, linear maps, XORs, counters, controlled phases, and uncomputation; no key-space phase table \\
gate-level S-DES check & small-width compilation test of the compute--phase--uncompute pattern, not a full-scale S-AES simulation \\
\bottomrule
\end{tabular}
\end{center}

\section{Numerical and Gate-Level Validation}
\label{sec:experiments}

\subsection{Protocol and fairness conditions}

The principal S-AES benchmark fixes
\begin{equation}
    (P_1,P_2,P_3)=(0x6F6B,0x1234,0xBEEF)
\end{equation}
and samples ten independent random 16-bit keys using the reproducible seed schedule generated from base seed $20260901$. Ciphertexts are generated from those keys and then treated as public. The Hybrid DDT and residual energy maps are fixed by the public S-box and public pairs. Adam is used to train the variational angles. No run is removed on the basis of its final probability, and the correct key is not supplied to the training objective.

The primary configuration uses the Hybrid staged schedule in Eq.~\eqref{eq:main-schedule-qip}: twelve Hybrid-DDT layers, three residual layers, and three global residual Hamming-shell layers. The DDT, residual, and shell optimization budgets are $1200$, $800$, and $1800$ iterations, respectively, with four warm-start restarts. The final probability is evaluated from the key-register state and the public-consistency set in Eq.~\eqref{eq:good-set}. A previously completed pairwise two-pair control is retained below as a reference for the effect of additional public data; it is not used as the primary Hybrid estimate.

\subsection{Ten-instance S-AES benchmark}

\begin{table}[t]
\centering
\caption{Ten independent random three-pair 16-bit S-AES instances for the principal Hybrid-DDT schedule. The final probability is the public-consistent-key probability evaluated on the key register; the last column gives the largest false-key probability among the reported top candidates.}
\label{tab:ten-qip}
\begin{tabular}{@{}rccc@{}}
\toprule
run & key & $p_{\Good}$ & largest false \\
\midrule
1 & $0x875F$ & 0.400682 & 0.006407 \\
2 & $0xFABF$ & 0.640035 & 0.008723 \\
3 & $0xB18A$ & 0.393665 & 0.005907 \\
4 & $0x1BEC$ & 0.424690 & 0.004540 \\
5 & $0xDF37$ & 0.732872 & 0.006388 \\
6 & $0x44F4$ & 0.540367 & 0.005174 \\
7 & $0xBA1F$ & 0.419556 & 0.005868 \\
8 & $0x57A6$ & 0.388892 & 0.003169 \\
9 & $0x3D20$ & 0.586158 & 0.004866 \\
10 & $0xD5BD$ & 0.770937 & 0.004824 \\
\midrule
mean & -- & 0.529785 & 0.005587 \\
\bottomrule
\end{tabular}
\end{table}

The correct key is ranked first in all ten runs. The correct-key probability ranges from $0.388892$ to $0.770937$, with mean $0.529785$ and population standard deviation $0.138956$. The largest false-key probability observed over the ten runs is $0.008723$, while its mean is $0.005587$. These values are concentration measurements on a reduced key register, not claims that a physical full-AES experiment succeeds with the same probability.

\subsection{Dependence on the number of public pairs}

\begin{table}[t]
\centering
\caption{Pairwise-DDT reference comparison between two- and three-pair S-AES inputs. The same ten random-key and optimizer-seed pairs are used in both conditions; the table is not the primary Hybrid benchmark.}
\label{tab:pairs-qip}
\scriptsize
\resizebox{\linewidth}{!}{%
\begin{tabular}{@{}lccccc@{}}
\toprule
public pairs & DDT components & mean $p_{\Good}$ & standard deviation & mean largest false & rank first \\
\midrule
$m=2$ & 4 & 0.208281 & 0.092611 & 0.017134 & $10/10$ \\
$m=3$ & 12 & 0.467769 & 0.193166 & 0.006050 & $10/10$ \\
\bottomrule
\end{tabular}
}
\end{table}

Within this pairwise reference, the three-pair mean is about $2.25$ times the two-pair mean, while the mean largest false-key probability decreases from $0.017134$ to $0.006050$. The increase is consistent with the intended qualitative role of additional public constraints: three public pairs provide three pairwise differences and twelve nibble-level DDT components, compared with one pairwise difference and four components for two pairs. It is an empirical comparison, not an unconditional theorem about all plaintext distributions, and it should not be conflated with the independent Hybrid estimate in Table~\ref{tab:ten-qip}.

\subsection{No-NPP pairwise and hybrid DDT ablation}

We also performed a separate matched ten-instance ablation with three public pairs. Both variants used the same fixed $8$-layer DDT budget, $3$ residual layers, no residual-shell stage, one Adam restart, and identical key and optimizer seeds within each pair. Only the DDT structure was changed: pairwise uses the twelve pairwise nibble components, whereas Hybrid augments them with four synchronized three-pair components. This is a no-NPP diagnostic ablation and is not the 18-layer principal benchmark. The Qiskit replay statistics are
\begin{table}[t]
\centering
\caption{Matched ten-instance comparison of pairwise and hybrid DDT structures. The reported probability is the public-consistent-key probability from the effective key-register replay.}
\label{tab:hybrid-qip}
\scriptsize
\resizebox{\linewidth}{!}{%
\begin{tabular}{@{}lccccc@{}}
\toprule
DDT structure & mean & minimum & maximum & standard deviation & correct key ranked first \\
\midrule
pairwise & 0.166258 & 0.120971 & 0.215272 & 0.037025 & $10/10$ \\
hybrid & 0.219947 & 0.171725 & 0.328364 & 0.052859 & $10/10$ \\
\midrule
hybrid $-$ pairwise & 0.053689 & $-0.041160$ & 0.171477 & 0.063420 & $7/10$ wins \\
\bottomrule
\end{tabular}
}
\end{table}

The hybrid construction has a mean absolute improvement of $0.053689$ (approximately $32.3\%$ relative to the pairwise mean) and wins seven of the ten paired trials. It is not uniformly better: pairwise wins three trials. The largest false-key probability in the reported sorted lists is at most $0.011384$ for pairwise and $0.008952$ for hybrid. These results support a distribution-level advantage for this fixed benchmark and budget; they do not establish dominance for other instances, plaintext ensembles, or optimized layer allocations.

\subsection{Gate-level S-DES validation}

We separately compile the diagnostic-oracle pattern for S-DES at gate level~\cite{SchaeferSDES}. The validation uses a 10-qubit key register, an 8-qubit cipher-work register, and one flag qubit, for 19 logical qubits in total. The S-box is realized as a reversible Boolean permutation and the public diagnostic is inserted through compute, controlled phase, and uncompute operations. No array indexed by the $2^{10}$ possible keys is loaded into the circuit. This experiment validates the reversible construction pattern and its Qiskit Aer execution; it is not used to infer the 16-bit S-AES concentration statistics.

\section{Validity and Limitations}
\label{sec:limits}

The deterministic completeness lemma is exact for the stated S-AES maps. The lower-tail proposition is conditional. A finite data set cannot prove Assumption~1; it can only support an empirical assessment by estimating the diagnostic-gap parameters $(\mu,\nu)$ on a calibration set and evaluating the tail prediction on held-out keys and held-out public instances. The training angles, energy maps, layer counts, and optimizer budgets must be fixed before the held-out evaluation. The current ten-run benchmark is an end-to-end performance sample; it should not be interpreted as a fitted proof of Assumption~1.

The following limitations define the claim precisely.
\begin{enumerate}[leftmargin=*]
    \item The cipher is 16-bit S-AES, not the 128-bit AES standard. The result demonstrates a mechanism on a controlled reduced cipher.
    \item The 16-bit numerical replay classically enumerates all candidate keys to construct effective diagonal phase arrays. The intended physical circuit uses reversible computation and does not require this enumeration, but that full S-AES gate-level circuit is not simulated here.
    \item The separation proposition does not prove finite-depth QAOA convergence. Optimizer behavior, circuit noise, routing, and shot cost remain empirical.
    \item The sub-Gaussian false-key gap assumption may fail for other ciphers, plaintext ensembles, or highly structured false-key families. Such failure is detectable through held-out tail tests.
    \item The method does not claim a practical attack on full AES.
\end{enumerate}

\section{Conclusion}
\label{sec:conclusion}

We have presented a structure-aware variational framework for quantum cryptanalysis of reduced block ciphers. The framework uses public S-box differential information and candidate-key-dependent inverse diagnostics to create reversible phase separators on a key register. Its theory separates deterministic completeness, statistical false-key separation, and finite-depth variational preparation. Its implementation separates the physically meaningful compute--phase--uncompute circuit from the compressed key-register replay used to audit a 16-bit benchmark.

On ten random three-pair S-AES instances, the structured schedule produced substantial key-register concentration and ranked the correct key first in every run. The result should be interpreted as evidence for a cipher-structure-to-quantum-circuit design principle, not as a full-AES attack. The next test is a held-out diagnostic-separation study and a width-controlled gate-level implementation in which the reversible oracle, shot-based training, and public verification are all executed without key-space enumeration.

\section*{Acknowledgments}
This work was supported by the National Natural Science Foundation of China (Grant Nos. 61871120 and 62071240), the Natural Science Foundation of Jiangsu Province (Grant Nos. BK20191259 and BK20220804), the Innovation Program for Quantum Science and Technology (Grant No. 2021ZD0302901), the Six Talent Peaks Project of Jiangsu Province (Grant No. XYDXX-003), and the Jiangsu Funding Program for Excellent Postdoctoral Talent (Grant No. 2022ZB107).

\bibliographystyle{unsrt}
\bibliography{prl_qaoa_aes_refs}

\end{document}